\documentclass[10pt,twocolumn]{article}

\usepackage[utf8]{inputenc}
\usepackage[T1]{fontenc}
\usepackage{latex8}
\usepackage{times}
\usepackage{amsmath}
\usepackage{amssymb}
\usepackage{amsthm}
\usepackage{color}
\usepackage{graphicx}
\usepackage{cancel}
\usepackage{xspace}
\usepackage[ruled,vlined,linesnumbered]{algorithm2e}

\renewcommand{\geq}{\geqslant}
\renewcommand{\leq}{\leqslant}
\renewcommand{\tilde}{\widetilde}

\newcommand{\defi}{\stackrel{\textbf{def}}{=}}
\newcommand{\lat}[1]{{\it #1}}

\newcommand{\fail}{\mathsf{Fail}}
\newcommand{\rr}{\mathsf{Eligible}}
\newcommand{\act}{\mathsf{Active}}
\newcommand{\run}{\mathsf{Run}}
\newcommand{\states}[1]{\ensuremath{\mathsf{States}\left(#1\right)}}
\newcommand{\Post}[1]{\ensuremath{\mathsf{Succ}\left(#1\right)}}
\newcommand{\MaxEltname}[1]{\ensuremath{\mathsf{Max}^{#1}}}
\newcommand{\MaxElt}[2]{\ensuremath{\MaxEltname{#1}\left(#2\right)}}
\newcommand{\tR}{\ensuremath{\tilde{R}}}
\newcommand{\Sbar}{\ensuremath{\overline{S}}}

\newcommand{\R}{\mathbin{R}}

\newcommand{\fsim}{\succcurlyeq}
\newcommand{\fsimstrict}{\succ}

\newcommand{\idlesim}{\succcurlyeq_{idle}}

\DeclareMathOperator{\nat}{nat}
\DeclareMathOperator{\rct}{rct}
\DeclareMathOperator{\lax}{laxity}

\DeclareMathOperator{\ttd}{ttd}

\theoremstyle{definition}
\newtheorem{definition}{Definition}
\theoremstyle{plain}
\newtheorem{lemma}[definition]{Lemma}
\newtheorem{theorem}[definition]{Theorem}

\author{Markus Lindström \hspace{3.5em} Gilles Geeraerts \hspace{3.5em} Joël Goossens \\\\
Université libre de Bruxelles\\
Département d'Informatique, Faculté des Sciences\\
Avenue Franklin D. Roosevelt 50, CP 212\\
1050 Bruxelles, Belgium\\
\{mlindstr, gilles.geeraerts, joel.goossens\}@ulb.ac.be
}
\title{A faster exact multiprocessor schedulability test for sporadic tasks}
\date{}

\begin{document}
\maketitle
\thispagestyle{empty}
\begin{abstract}
  Baker and Cirinei introduced an exact but naive
  algorithm~\cite{Baker2007}, based on solving a state reachability
  problem in a finite automaton, to check whether sets of sporadic
  hard real-time tasks are schedulable on identical multiprocessor
  platforms. However, the algorithm suffered from poor performance due
  to the exponential size of the automaton relative to the size of the
  task set. In this paper, we successfully apply techniques developed
  by the formal verification community, specifically antichain
  algorithms \cite{Doyen2010}, by defining and proving the correctness
  of a simulation relation on Baker and Cirinei's automaton. We show
  our improved algorithm yields dramatically improved performance for
  the schedulability test and opens for many further improvements.
\end{abstract}

\Section{Introduction}
\label{introduction}

In this research we consider the schedulability problem of hard real-time sporadic constrained deadline task systems upon identical multiprocessor platforms. Hard real-time systems are systems where tasks are not only required to provide correct computations but are also require to adhere to strict deadlines \cite{Liu1973}.

Devising an exact schedulability criterion for sporadic task sets on multiprocessor platforms has so far proven difficult due to the fact that there is no known worst case scenario (nor critical instant). It was notably shown in~\cite{Goossens2002} that the periodic case is not necessarily the worst on multiprocessor systems. In this context, the real-time community has mainly been focused on the development of \emph{sufficient} schedulability tests that correctly identify all unschedulable task sets, but may misidentify some schedulable systems as being unschedulable~\cite{Baker2006} using a given platform and scheduling policy (see e.g.~\cite{Bertogna2011, Baruah2007}).

Baker and Cirinei introduced the first correct algorithm~\cite{Baker2007} that verified \emph{exactly} whether a sporadic task system was schedulable on an identical multiprocessor platform by solving a reachability problem on a finite state automaton using a naive brute-force algorithm, but it suffered from the fact that the number of states was exponential in the size of the task sets and its periods, which made the algorithm intractable even for small task sets with large enough periods.

In this paper, we apply techniques developed by the formal verification community, specifically Doyen, Raskin \lat{et al.}~\cite{Doyen2010, DeWulf2006} who developed faster algorithms to solve the reachability problem using algorithms based on data structures known as {\em antichains}. Their method has been shown to be provably better~\cite{Doyen2010} than naive state traversal algorithms such as those used in~\cite{Baker2007} for deciding reachability from a set of initial states to a given set of final states.

An objective of this work is to be as self-contained as possible to allow readers from the real-time community to be able to fully understand the concepts borrowed from the formal verification community. We also hope our work will kickstart a ``specialisation'' of the methods presented herein within the realm of real-time scheduling, thus bridging the two communities.

\paragraph{Related work.} This work is not the first contribution to
apply techniques and models first proposed in the setting of formal
verification to real-time scheduling. In the field of operational
research, Abdeddaïm and Maler have studied the use of stopwatch
automata to solve job-shop scheduling problems
\cite{Maler2002}. Cassez has recently exploited game theory,
specifically timed games, to bound worst-case execution times on
modern computer architectures, taking into account caching and
pipelining~\cite{Cassez2011}. Fersman \lat{et al.}\ have studied a similar
problem and introduced task automata which assume continuous time \cite{Fersman2007},
whereas we consider discrete time in our work. They showed that, given selected constraints, schedulability could be undecidable in their model. Bonifaci and Marchetti-Spaccamela have studied the related problem of feasibility of multiprocessor sporadic systems in \cite{Bonifaci2010} and have established an upper bound on its complexity.


\paragraph{This research.}

We define a restriction to constrained deadlines (systems where the relative deadline of tasks is no longer than their minimal interarrival time) of Baker and Cirinei's automaton in a more formal way than in~\cite{Baker2007}. We also formulate various scheduling policy properties in the framework of this automaton such as memorylessness.

Our main contribution is the design and proof of correctness of a non-trivial \emph{simulation relation} on the automaton, required to successfully apply a generic algorithm developed in the formal verification community, known as an \emph{antichain algorithm} to Baker and Cirinei's automaton to prove or disprove the schedulability of a given sporadic task system.

Finally, we will show through implementation and experimental analysis that our proposed algorithm outperforms Baker and Cirinei's original brute-force algorithm.

\paragraph{Paper organization.}

Section~\ref{section:problem} defines the real-time scheduling problem we are focusing on, i.e.\ devising an exact schedulability test for sporadic task sets on identical multiprocessor platforms. Section~\ref{section:automaton} will formalize the model (a non-deterministic automaton) we will use to describe the problem and we formulate how the schedulability test can be mapped to a reachability problem in this model. We also formalize various real-time scheduling concepts in the framework of our formal model.

Section~\ref{section:reachability} then discusses how the reachability problem can be solved. We present the classical breadth-first algorithm used in~\cite{Baker2007} and we introduce an improved algorithm that makes use of techniques borrowed from the formal verification community~\cite{Doyen2010}. The algorithm requires coarse \emph{simulation relations} to work faster than the standard breadth-first algorithm. Section~\ref{section:idlesim} introduces the \emph{idle tasks simulation relation} which can be exploited by the aforementioned algorithm. 

Section~\ref{section:experimental} then showcases experimental results comparing the breadth-first and our improved algorithm using the aforementioned simulation relation, showing that our algorithm outperforms the naive one. Section~\ref{section:conclusions} concludes our work. Appendix~\ref{appendix:lemma} gives a detailed proof of a lemma we use in Section~\ref{section:reachability}.

\Section{Problem definition} 
\label{section:problem}

We consider an identical multiprocessor platform with $m$ processors and a sporadic task set $\tau = \{ \tau_1, \tau_2, \ldots, \tau_n \}$. Time is assumed to be discrete. A sporadic task $\tau_i$ is characterized by a {\em minimum interarrival time} $T_i > 0$, a {\em relative deadline} $D_i > 0$ and a {\em worst-case execution time} (also written WCET) $C_i > 0$. A sporadic task $\tau_i$ submits a potentially infinite number of jobs to the system, with each request being separated by at least $T_i$ units of time. We will assume jobs are not parallel, i.e.\ only execute on one single processor (though it may migrate from a processor to another during execution). We also assume jobs are independent. We wish to establish an \emph{exact} schedulability test for any sporadic task set $\tau$ that tells us whether the set is schedulable on the platform with a given deterministic, predictable and preemptive scheduling policy. In the remainder of this paper, we will assume we only work with {\em constrained deadline} systems (i.e.\ where $\forall \tau_i : D_i \leq T_i$) which embody many real-time systems in practice.

\Section{Formal definition of the Baker-Cirinei automaton}
\label{section:automaton}

Baker and Cirinei's automaton as presented in~\cite{Baker2007} models the evolution of an \emph{arbitrary}
deadline sporadic task set (with a FIFO policy for jobs of a given task) scheduled on an identical multiprocessor
platform with $m$ processors. In this paper, we focus on constrained
deadline systems as this hypothesis simplifies the definition of the
automaton. We expect to analyze
Baker and Cirinei's more complete construct in future
works.

The model presented herein allows use of \emph{preemptive}, \emph{deterministic} and \emph{predictable} scheduling policies. It can, however, be generalized to model broader classes of schedulers. We will discuss this aspect briefly in Section~\ref{section:conclusions}.

\begin{definition}
  An \emph{automaton} is a tuple $A=\langle V, E, S_0, F\rangle$,
  where $V$ is a finite set of \emph{states}, $E\subseteq V\times V$ is
  the set of \emph{transitions}, $S_0\in V$ is the \emph{initial state}
  and $F\subseteq V$ is a set of \emph{target states}.
\end{definition}

The problem on automata we are concerned with is that of
\emph{reachability} (of target states).  A \emph{path} in an automaton
$A=\langle V, E, S_0, F\rangle$ is a finite sequence
$v_1,\ldots,v_\ell$ of states s.t.\  for all $1\leq i\leq \ell-1$:
$(v_i,v_{i+1})\in E$. Let $V'\subseteq V$ be a set of states of $A$. If
there exists a path $v_1,\ldots,v_\ell$ in $A$ s.t.\  $v_\ell\in V'$, we
say that \emph{$v_1$ can reach $V'$}. Then, the \emph{reachability
  problem} asks, given an automaton $A$ whether the initial state
$S_0$ can reach the set of target states $F$.

Let $\tau = \{ \tau_1, \tau_2, \ldots, \tau_n \}$ be a set of sporadic tasks and $m$ be a number of processors. This section is devoted to
explaining how to model the behaviour of such a system by means of an
automaton $A$, and how to reduce the schedulability problem of $\tau$ on
$m$ processors to an instance of the reachability problem in $A$. At
any moment during the execution of such a system, the information we
need to retain about each task $\tau_i$ are: $(i)$ the \emph{earliest
  next arrival time} $\nat(\tau_i)$ relative to the current instant
and $(ii)$ the remaining processing time $\rct(\tau_i)$ of the
currently ready job of $\tau_i$. Hence the definition of \emph{system
  state}:

\begin{definition}[System states]
\label{def:nat+rct}
Let $\tau = \{ \tau_1, \tau_2, \ldots, \tau_n \}$ be a set of sporadic tasks.  A \emph{system state} of $\tau$ is a tuple
$S=\langle\nat_S,\rct_S\rangle$ where $\nat_S$ is a function from
$\tau$ to $\{0,1,\ldots T_{\max}\}$ where $T_{\max}\defi\max_i T_i$, and $\rct_S$ is a
function from $\tau$ to $\{0,1,\ldots, C_{\max}\}$, where $C_{\max} \defi \max_iC_i$. We
denote by $\states{\tau}$ the set of all system states of $\tau$.
\end{definition}

In order to define the set of transitions of the automaton, we need to
rely on ancillary notions:
\begin{definition}[Eligible task]
\label{def:eligible}
A task $\tau_i$ is \emph{eligible} in the state $S$ if it can submit a
job (i.e.\ if and only if the task does not currently have an active
job and the last job was submitted at least $T_i$ time units ago) from
this configuration. Formally, the set of eligible tasks in state $S$ is: 
\begin{eqnarray*}
\rr(S)&\defi&\{\tau_i\mid \nat_S(\tau_i) = \rct_S(\tau_i) = 0\}
\end{eqnarray*}
\end{definition}

\begin{definition}[Active task]
\label{def:active}
A task is \emph{active} in state $S$ if it currently has a job that
has not finished in $S$. Formally, the set of active tasks in $S$ is:
\begin{eqnarray*}
  \act(S)&\defi&\{\tau_i\mid \rct_S(\tau_i) > 0\}
\end{eqnarray*}
A task that is not active in $S$ is said to be {\em idle} in $S$.
\end{definition}




\begin{definition}[Laxity \cite{Baker2007}]
\label{def:laxity}
The laxity of a task $\tau_i$ in a system state $S$ is:
\[
\lax_S(\tau_i) \defi\nat_S(\tau_i) - (T_i - D_i) -\rct_S(\tau_i)
\]
\end{definition}

\begin{definition}[Failure state]
\label{def:fail}
A state $S$ is a \emph{failure state} iff the laxity of at least one
task is negative in $S$. Formally, the set of failure states on $\tau$ is:
\begin{eqnarray*}
  \fail_\tau&\defi&\{S\mid\exists \tau_i\in\tau:\lax_S(\tau_i) < 0\}
\end{eqnarray*}
\end{definition}

Thanks to these notions we are now ready to explain how to build the
transition relation of the automaton that models the behaviour of
$\tau$. For that purpose, we first choose a
\emph{scheduler}. Intuitively, a scheduler is a
\emph{function}\footnote{Remark that by modeling the scheduler as a
  function, we restrict ourselves to \emph{deterministic schedulers}.}
$\run$ that maps each state $S$ to a set of at most $m$ active tasks
$\run(S)$ to be run:
\begin{definition}[Scheduler]\label{def:sched}
  A (deterministic) \emph{scheduler} for $\tau$ on $m$ processors is
  a function $\run:\states{\tau} \to 2^\tau$ s.t.\ for all $S$:
  $\run(S)\subseteq \act(S)$ and $0\leq |\run(S)|\leq m$. Moreover:
  \begin{enumerate}\item 
    $\run$ is \emph{work-conserving} iff for all $S$, $|\run(S)|=
    \min\{m, |\act(S)|\}$
  \item $\run$ is \emph{memoryless} iff for all
    $S_1,S_2\in\states{\tau}$ with $\act(S_1)=\act(S_2)$:
    $$
    \begin{array}{c}
        \forall \tau_i \in \act(S_1) : 
        \left( \begin{array}{l}
            \nat_{S_1}(\tau_i) =\nat_{S_2}(\tau_i)\\  
            \land\rct_{S_1}(\tau_i) =\rct_{S_2}(\tau_i) \end{array} 
        \right)\\
        \textrm{implies } \run(S_1) = \run(S_2)
      \end{array}
$$
  \end{enumerate}
\end{definition}
Intuitively, the work-conserving property implies that the scheduler
always exploits as many processors as available. The memoryless
property implies that the decisions of the scheduler are not affected
by tasks that are idle and that the scheduler does not consider
the past to make its decisions.

As examples, we can formally define the preemptive global DM and EDF schedulers.

\begin{definition}[Preemptive global DM scheduler] Let $\ell \defi \min\{m,|\act(S)|\}$. Then, $\run_{\text{DM}}$ is a function that computes $\run_{\text{DM}}(S) \defi \{\tau_{i_1}, \tau_{i_2}, \ldots, \tau_{i_\ell}\}$ s.t. for all $1 \leq j \leq \ell$ and for all $\tau_k$ in $\act(S) \setminus \run_{\text{DM}}(S)$, we have $D_k > D_{i_j}$ or $D_k = D_{i_j} \land k > i_j$.
\end{definition}

\begin{definition}[Preemptive global EDF scheduler] Let $\ttd_S(\tau_i) \defi \nat_S(\tau_i) - (T_i-D_i)$ be the time remaining before the absolute deadline of the last submitted job~\cite{Baker2007} of $\tau_i \in \act(S)$ in state $S$. Let $\ell \defi \min\{m,|\act(S)|\}$. Then, $\run_{\text{EDF}}$ is a function that computes $\run_{\text{EDF}}(S) \defi \{\tau_{i_1}, \tau_{i_2}, \ldots, \tau_{i_\ell}\}$ s.t.\ for all $1 \leq j \leq \ell$ and for all $\tau_k$ in $\act(S) \setminus \run_{\text{EDF}}(S)$, we have $\ttd_S(\tau_k) > \ttd_S(\tau_{i_j})$ or $\ttd_S(\tau_k) = \ttd_S(\tau_{i_j}) \land k > i_j$.
\end{definition}

By Definition~\ref{def:sched}, global DM and EDF are thus work-conserving and it can also be verified that they are memoryless. In \cite{Baker2007}, suggestions to model several other schedulers were presented. It was particularily shown that adding supplementary information to system states could allow broader classes of schedulers to be used. Intuitively, states could e.g.\ keep track of what tasks were executed in their predecessor to implement non-preemptive schedulers.

Clearly, in the case of the scheduling of sporadic tasks, two types of
events can modify the current state of the system:
\begin{enumerate}
\item {\em Clock-tick transitions} model the elapsing of time for one
  time unit, i.e.\ the execution of the scheduler and the running of
  jobs.
\item {\em Request transitions} (called \emph{ready transitions}
  in~\cite{Baker2007}) model requests from sporadic tasks at a given
  instant in time.
\end{enumerate}

Let $S$ be a state in $\states{\tau}$, and let $\run$ be a
scheduler. Then, letting one time unit elapse from $S$ under the
scheduling policy imposed by $\run$ amounts to decrementing the $\rct$
of the tasks in $\run(S)$ (and only those tasks), and to decrementing
the $\nat$ of all tasks. Formally:
\begin{definition}\label{def:clocktick}
  Let $S=\langle\nat_S,\rct_S\rangle\in\states{\tau}$ be a system
  state and $\run$ be a scheduler for $\tau$ on $m$ processors. Then,
  we say that $S^+=\langle\nat_S^+,\rct_S^+\rangle$ is a \emph{clock-tick
  successor} of $S$ under $\run$, denoted
  $S\xrightarrow{\run}S^+$ iff:
  \begin{enumerate}
  \item for all $\tau_i\in\run(S)$: $\rct_S^+(\tau_i)=\rct_S(\tau_i)-1$ ;
  \item for all $\tau_i\not\in\run(S)$: $\rct_S^+(\tau_i)=\rct_S(\tau_i)$ ;
  \item for all $\tau_i\in\tau$:
    $\nat_S^+(\tau_i)=\max\{\nat_S(\tau_i)-1,0\}$.
  \end{enumerate}
\end{definition}

Let $S$ be a state in $\states{\tau}$. Intuitively, when the system is
in state $S$, a request by some task $\tau_i$ for submitting a new job
has the effect to update $S$ by setting $\nat(\tau_i)$ to $T_i$ and
$\rct(\tau_i)$ to $C_i$. This can be generalised to sets of
tasks. Formally:

\begin{definition}\label{def:reqtrans}
  Let $S\in\states{\tau}$ be a system state and let
  $\tau'\subseteq\rr(S)$ be a set of tasks that are eligible to submit
  a new job in the system. Then, we say that $S'$ is a \emph{$\tau'$-request
  successor} of $S$, denoted $S\xrightarrow{\tau'}S'$, iff:
  \begin{enumerate}
  \item for all $\tau_i\in \tau'$: $\nat_{S'}(\tau_i)=T_i$ and
    $\rct_{S'}(\tau_i)=C_i$
  \item for all $\tau_i\in\tau\setminus\tau'$:
    $\nat_{S'}(\tau_i)=\nat_S(\tau_i)$ and
    $\rct_{S'}(\tau_i)=\rct_S(\tau_i)$.
  \end{enumerate}
\end{definition}
Remark that we allow $\tau'=\emptyset$ (that is, no task asks to
submit a new job in the system).

We are now ready to define the automaton $A(\tau, \run)$ that
formalises the behavior of the system of sporadic tasks $\tau$, when
executed upon $m$ processors under a scheduling policy $\run$:
\begin{definition}\label{def:automaton}
  Given a set of sporadic tasks $\tau$ and a scheduler $\run$ for
  $\tau$ on $m$ processors, the automaton $A(\tau, \run)$ is the
  tuple $\langle V, E, S_0,F\rangle$ where:
  \begin{enumerate}
  \item $V=\states{\tau}$
  \item $(S_1,S_2)\in E$ iff there exists $S'\in \states{\tau}$ and
    $\tau'\subseteq \tau$
    s.t.\ $S_1\xrightarrow{\tau'}S'\xrightarrow{\run}S_2$.
  \item $S_0=\langle\nat_0,\rct_0\rangle$ where for all
    $\tau_i\in\tau$, $\nat_0(\tau_i)=\rct_0(\tau_i)=0$.
  \item $F=\fail_\tau$
  \end{enumerate}
\end{definition}

Figure~\ref{fig:simulation} illustrates a possible graphical
representation of one such automaton, which will be 
analyzed further in Section~\ref{section:idlesim}.
On this example, the automaton depicts the following EDF-schedulable
sporadic task set using an EDF scheduler and assuming $m=2$:

\[
\begin{array}{c|ccc}
& T_i & D_i & C_i\\\hline
\tau_1 & 2 & 2 & 1\\
\tau_2 & 3 & 3 & 2
\end{array}
\]

System states are represented by nodes. For the purpose of saving space,
we represent a state $S$ with the $[\alpha\beta,\gamma\delta]$
format, meaning $\nat_S(\tau_1) = \alpha$, $\rct_S(\tau_1) = \beta$,
$\nat_S(\tau_2) = \gamma$ and $\rct_S(\tau_2) = \delta$.  We explicitly
represent clock-tick transitions by edges labelled with $\run$, and
$\tau'$-request transitions by edges labelled with $\tau'$. $\tau'=\emptyset$
loops are implicit on each state. Note that, in accordance with
Definition~\ref{def:automaton}, there are no successive $\tau'$-request
transitions, and there are thus no such transitions from states such as
$[21,00]$ and $[00,32]$. Also note that the automaton indeed models the
evolution of a sporadic system, of which the periodic case is one possible
path (the particular case of a synchronous system is found by taking the maximal $\tau'$-request transition whenever possible, starting from $[00,00]$).

We remark that our definition deviates slightly from that of Baker and
Cirinei. In our definition, a path in the automaton corresponds to an
execution of the system that alternates between requests transitions
(possibly with an empty set of requests) and clock-tick
transitions. In their work \cite{Baker2007}, Baker and Cirinei allow
any sequence of clock ticks and requests, but restrict each request to
a single task at a time. It is easy to see that these two definitions
are equivalent. A sequence of $k$ clock ticks in Baker's automaton
corresponds in our case to a path $S_1,S_2,\ldots S_{k+1}$ s.t.\ for
all $1\leq i\leq k$:
$S_i\xrightarrow{\emptyset}S_i\xrightarrow{\run}S_{i+1}$. A maximal
sequence of successive requests by $\tau_1,\tau_2,\ldots,\tau_k$,
followed by a clock tick corresponds in our case to a single edge
$(S_1,S_2)$
s.t.\ $S_1\xrightarrow{\{\tau_1,\ldots,\tau_k\}}S'\xrightarrow{\run}S_2$
for some $S'$. Conversely, each edge $(S_1,S_2)$ in $A(\tau, \run)$
s.t.\ $S_1\xrightarrow{\tau'}S'\xrightarrow{\run}S_2$, for some state
$S'$ and set of tasks $\tau'=\{\tau_1,\ldots,\tau_k\}$, corresponds to
a sequence of successive requests\footnote{Remark that the order does
  not matter.} by $\tau_1$,\ldots, $\tau_k$ followed by a clock tick
in Baker's setting.


The purpose of the definition of $A(\tau,\run)$ should now be clear to
the reader. Each possible execution of the system corresponds to a
path in $A(\tau,\run)$ and vice-versa. States in $\fail_\tau$
correspond to states of the system where a deadline will unavoidably
be missed. Hence, \emph{the set of sporadic tasks $\tau$ is feasible
  under scheduler $\run$ on $m$ processors iff $\fail_\tau$ is not
  reachable in $A(\tau,\run)$} \cite{Baker2007}. Unfortunately, the
number of states of $A(\tau,\run)$ can be intractable even for very small sets of
tasks $\tau$. In the next section we present generic techniques to
solve the reachability problem in an efficient fashion, and apply them
to our case. Experimental results given in
Section~\ref{section:experimental} demonstrate the practical interest of
these methods.

\Section{Solving the reachability problem}
\label{section:reachability}

Let us now discuss techniques to solve the reachability problem. Let
$A=\langle V, E, S_0,F\rangle$ be an automaton. For any $S\in V$, let
$\Post{S}=\{S'\mid (S,S')\in E\}$ be the set of one-step successors of
$S$. For a set of states $R$, we let $\Post{R}=\cup_{S\in
  R}\Post{S}$. Then, solving the reachability problem on $A$ can be
done by a \emph{breadth-first traversal} of the automaton, as shown in
Algorithm~\ref{algo:bf}.

\begin{algorithm}
  \Begin {
    $i\leftarrow 0$ \;
    $R_0\leftarrow \{S_0\}$ \;
    \Repeat{$R_i=R_{i-1}$}{
      $i\leftarrow i+1$ \;
      $R_i\leftarrow R_{i-1}\cup\Post{R_{i-1}}$ \;
      \lIf{$R_i\cap F\neq\emptyset$}{ \Return{\texttt{Reachable}} \;}
    }
    \Return{\texttt{Not reachable}} \;
  }
  \caption{Breadth-first traversal.\label{algo:bf}}
\end{algorithm}
Intuitively, for all $i\geq 0$, $R_i$ is the set of states that are
reachable from $S_0$ in $i$ steps at most. The algorithm computes the
sets $R_i$ up to the point where $(i)$ either a state from $F$ is met
or $(ii)$ the sequence of $R_i$ stabilises because no new states have
been discovered, and we declare $F$ to be unreachable. This algorithm
always terminates and returns the correct answer. Indeed, either $F$
is reachable in, say $k$ steps, and then $R_k\cap F\neq \emptyset$,
and we return `\texttt{Reachable}'. Or $F$ is not reachable, and the
sequence eventually stabilises because $R_0\subseteq R_1\subseteq
R_2\subseteq\cdots\subseteq V$, and $V$ is a finite set. Then, we exit
the loop and return `\texttt{Not reachable}'. Remark that this
algorithm has the advantage that the whole automaton does not need be
stored in memory before starting the computation, as
Definition~\ref{def:clocktick} and Definition~\ref{def:reqtrans} allow
us to compute $\Post{S}$ \emph{on the fly} for any state
$S$. Nevertheless, in the worst case, this procedure needs to explore
the whole automaton and is thus in $\mathcal{O}(|V|)$ which can be too
large to handle in practice \cite{Baker2007}.

Equipped with such a simple definition of \emph{automaton}, this is
the best algorithm we can hope for. However, in many practical cases,
the set of states of the automaton is endowed with a \emph{strong semantic} that can be exploited to speed up
Algorithm~\ref{algo:bf}. In our case, states are tuples of integers
that characterise sporadic tasks running in a system. To harness this
information, we rely on the formal notion of \emph{simulation}:

\begin{definition}\label{def:simu}
  Let $A=\langle V, E, S_0,F\rangle$ be an automaton. A \emph{simulation
    relation} for $A$ is a preorder $\fsim\subseteq V\times V$ s.t.:
  \begin{enumerate}
  \item For all $S_1$, $S_2$, $S_3$ s.t.\ $(S_1,S_2)\in E$ and
    $S_3\fsim S_1$, there exists $S_4$ s.t.\ $(S_3,S_4)\in E$ and
    $S_4 \fsim S_2$.
  \item For all $S_1$, $S_2$ s.t.\ $S_1\fsim S_2$: $S_2\in F$
    implies $S_1\in F$.
  \end{enumerate}
  Whenever $S_1\fsim S_2$, we say that $S_1$ \emph{simulates}
  $S_2$. Whenever $S_1\fsim S_2$ but $S_2\not\fsim S_1$, we write
  $S_1\fsimstrict S_2$.
\end{definition}
Intuitively, this definition says that whenever a state $S_3$ simulates
a state $S_1$, then $S_3$ can \emph{mimick} every possible move of
$S_1$ by moving to a similar state: for every edge $(S_1,S_2)$, there
is a corresponding edge $(S_3,S_4)$, where $S_4$ simulates
$S_2$. Moreover, we request that a \emph{target state} can only be
simulated by a target state. Remark that for a given automaton there
can be several simulation relations (for instance, equality is always
a simulation relation). 

The key consequence of this definition is that \textbf{if} $S_2$ is a
state that can reach $F$, and if $S_1\fsim S_2$ \textbf{then}
\emph{$S_1$ can reach $F$ too}. Indeed, if $S_2$ can reach $F$, there
is a path $v_0,v_1,\ldots, v_n$ with $v_0=S_2$ and $v_n\in F$. Using
Definition~\ref{def:simu} we can inductively build a path
$v_0',v_1',\ldots, v_n'$ s.t.\ $v_0'=S_1$ and $v_i'\fsim v_i$ for all
$i\geq 0$. Thus, in particular $v_n'\fsim v_n\in F$, hence $v_n'\in F$
by Definition~\ref{def:simu}. This means that $S_1$ can reach $F$ too.
Thus, when we compute two states $S_1$ and $S_2$ with $S_1\fsim S_2$,
at some step of Algorithm~\ref{algo:bf}, we \emph{do not need to
  further explore the successors of $S_2$}. Indeed,
Algorithm~\ref{algo:bf} tries to detect reachable target states. So,
if $S_2$ cannot reach a failure state, it is safe not to explore its
succesors. Otherwise, if $S_2$ \emph{can} reach a target state, then 
$S_1$ can reach a target state too, so it is safe to explore the
successors of $S_1$ only. By exploiting this heuristic,
Algorithm~\ref{algo:bf} could explore only a (small) subset of the
states of $A$, which has the potential for a dramatic improvement in
computation time. Remark that such techniques have already been
exploited in the setting of \emph{formal verification}, where several
so-called \emph{antichains algorithms} have been studied
\cite{DeWulf2006,Doyen2010,DBLP:conf/cav/FiliotJR09} and have proved
to be \emph{several order of magnitudes more efficient} than the
classical techniques of the literature.

Formally, for a set of states $V'\subseteq V$, we let
$\MaxElt{\fsim}{V'}=\{S\in V'\mid \nexists S'\in V' \textrm{ with }
S'\fsimstrict S\}$.  Intuitively, $\MaxElt{\fsim}{V'}$ is obtained
from $V'$ by removing all the states that are simulated by some other
state in $V'$. So the states we keep in $\MaxElt{\fsim}{V'}$ are
irredundant\footnote{They form an \emph{antichain} of states wrt
  $\fsim$.} wrt $\fsim$.  Then, we consider
Algorithm~\ref{algo:bfanti} which is an improved version of
Algorithm~\ref{algo:bf}.
\begin{algorithm}
  \Begin {
    $i\leftarrow 0$ \;
    $\tR_0\leftarrow \{S_0\}$ \;
    \Repeat{$\tR_i=\tR_{i-1}$}{
      $i\leftarrow i+1$ \;
      $\tR_i\leftarrow \tR_{i-1}\cup\Post{\tR_{i-1}}$ \;
      $\tR_i\leftarrow \MaxElt{\fsim}{\tR_i}$ \;
      \lIf{$\tR_i\cap F\neq\emptyset$}{ \Return{\texttt{Reachable}} \;}
    }
    \Return{\texttt{Not reachable}} \;
  }
  \caption{Improved breadth-first traversal.\label{algo:bfanti}}
\end{algorithm}

Proving the correctness and termination of Algorithm~\ref{algo:bfanti}
is a little bit more involved than for Algorithm~\ref{algo:bf} and
relies on the following lemma (proof in appendix):
\begin{lemma}\label{lem:prop-bf-anti}
  Let $A$ be an automaton and let $\fsim$ be a simulation relation for
  $A$. Let $R_0,R_1,\ldots$ and $\tR_0,\tR_1,\ldots$ denote
  respectively the sequence of sets computed by
  Algorithm~\ref{algo:bf} and Algorithm~\ref{algo:bfanti} on
  $A$. Then, for all $i\geq 0$: $\tR_i=\MaxElt{\fsim}{R_i}$.
\end{lemma}
Intuitively, this means that some state $S$ that is in $R_i$ could not
be present in $\tR_i$, but that we always keep in $\tR_i$ a state $S'$
that simulates $S$.  Then, we can prove that:
\begin{theorem}
  For all automata $A=\langle V,E,S_0,F\rangle$,
  Algorithm~\ref{algo:bfanti} terminates and returns
  \texttt{``Reachable''} iff $F$ is reachable in $A$.
\end{theorem}
\begin{proof}
  The proof relies on the comparison between the sequence of sets
  $R_0,R_1,\ldots$ computed by Algorithm~\ref{algo:bf} (which is
  correct and terminates) and the sequence $\tR_0,\tR_1,\ldots$
  computed by Algorithm~\ref{algo:bfanti}. 

  Assume $F$ is reachable in $A$ in $k$ steps and not reachable in
  less than $k$ steps. Then, there exists a path $v_0, v_1,\ldots v_k$
  with $v_0=S_0$, $v_k\in F$, and, for all $0\leq i\leq k$ $v_i\in
  R_k$. Let us first show \textit{per absurdum} that the loop in
  Algorithm~\ref{algo:bfanti} does not finish before the $k$th
  step. Assume it is not the case, i.e.\ there exists $0< \ell< k$
  s.t.\ $\tR_\ell=\tR_{\ell-1}$. This implies that
  $\MaxElt{\fsim}{R_\ell}=\MaxElt{\fsim}{R_{\ell-1}}$ through Lemma~\ref{lem:prop-bf-anti}. Since
  $R_\ell\neq \R_{\ell-1}$, we deduce that all the states that have
  been added to $R_\ell$ are simulated by some state already present
  in $R_{\ell-1}$: for all $S\in R_\ell$, there is $S'\in R_{\ell-1}$
  s.t.\ $S'\fsim S$. Thus, in particular, there is $S'\in R_{\ell-1}$
  s.t.\ $S'\fsim v_\ell$. We consider two cases. Either there is
  $S'\in R_{\ell-1}$ s.t.\ $S'\fsim v_k$. Since $v_k\in F$, $F\cap
  R_{\ell-1}\neq \emptyset$, which contradicts our hypothesis that $F$
  is not reachable in less than $k$ steps. Otherwise, let $0\leq m< k$
  be the least position in the path s.t.\ there is $S'\in R_{\ell-1}$
  with $S'\fsim v_m$, but there is no $S''\in R_{\ell-1}$ with
  $S''\fsim v_{m+1}$. In this case, since $S'\fsim v_m$ and
  $(v_m,v_{m+1})\in E$, there is $S\in\Post{S'}\subseteq R_{\ell}$
  s.t.\ $S\fsim v_{m+1}$. However, we have made the hypothesis that
  every element in $R_\ell$ is simulated by some element in
  $R_{\ell-1}$. Thus, there is $S'' \in R_{\ell-1}$ s.t.\ $S''\fsim
  S$. Since $S\fsim v_{m+1}$, we deduce that $S''\fsim v_{m+1}$, with
  $S''\in R_{\ell-1}$, which contradicts our assumption that $S'' \notin R_{\ell-1}$. Thus,
  Algorithm~\ref{algo:bfanti} will not stop before the $k$th
  iteration, and we know that there is $S_F\in R_k$ s.t.\ $S_F\in
  F$.  By Lemma~\ref{lem:prop-bf-anti}, $\tR_k=\MaxElt{\fsim}{R_k}$,
  hence there is $S'\in \tR_k$ s.t.\ $S'\fsim S$. By
  Definition~\ref{def:simu}, $S'\in F$ since $S\in F$. Hence,
  $\tR_k\cap F\neq \emptyset$ and Algorithm~\ref{algo:bfanti}
  terminates after $k$ steps with the correct answer.

  Otherwise, assume $F$ is not reachable in $A$. Hence, for every
  $i\geq 0$, $R_i\cap F=\emptyset$. Since $\tR_i\subseteq R_i$ for all
  $i\geq 0$, we conclude that $\tR_i\cap F=\emptyset$ for all $i\geq
  0$. Hence, Algorithm~\ref{algo:bfanti} never returns
  \texttt{``Reachable''} in this case. It remains to show that the
  \textbf{repeat} loop eventually terminates. Since $F$ is not
  reachable in $A$, there is $k$ s.t.\ $R_k=R_{k-1}$. Hence,
  $\MaxElt{\fsim}{R_k}=\MaxElt{\fsim}{R_{k-1}}$. By
  Lemma~\ref{lem:prop-bf-anti} this implies that
  $\tR_k=\tR_{k-1}$. Thus, Algorithm~\ref{algo:bfanti} finishes after
  $k$ steps and returns \texttt{``Not reachable''}.
\end{proof}

In order to apply Algorithm~\ref{algo:bfanti}, it remains to show how
to compute a simulation relation, which should contain as many pairs of
states as possible, since this raises the chances to avoid exploring
some states during the breadth-first search. It is well-known that the
largest simulation relation of an automaton can be computed in
polynomial time wrt the size of the automaton
\cite{DBLP:conf/focs/HenzingerHK95}. However, this requires first
computing the whole automaton, which is exactly what we want to avoid in
our case. So we need to define simulations relations that can be
computed \textit{a priori}, only by considering the structure of the
states (in our case, the functions $\nat$ and $\rct$). This is the
purpose of the next section.

\Section{Idle tasks simulation relation}
\label{section:idlesim}

In this section we define a simulation relation $\idlesim$, called the
\emph{idle tasks simulation relation} that can be computed by
inspecting the values $\nat$ and $\rct$ stored in the states.

\begin{definition}
  Let $\tau$ be a set of sporadic tasks. Then, the {\em idle tasks
    preorder} $\idlesim\subseteq\states{\tau}\times\states{\tau}$ is
  s.t.\ for all $S_1$,$S_2$: $S_1 \idlesim S_2$ iff
  \begin{enumerate}
  \item $\rct_{S_1}=\rct_{S_2}$ ;
  \item for all $\tau_i$ s.t.\ $\rct_{S_1}(\tau_i) = 0$:
    $\nat_{S_1}(\tau_i) \leq\nat_{S_2}(\tau_i)$~;
  \item for all $\tau_i$ s.t.\ $\rct_{S_1}(\tau_i) > 0$:
    $\nat_{S_1}(\tau_i) =\nat_{S_2}(\tau_i)$.
  \end{enumerate}

  
\end{definition}
Notice the relation is reflexive as well as transitive, and thus
indeed a preorder. It also defines a partial order on $\states{\tau}$
as it is antisymmetric. Moreover, since $S_1\idlesim S_2$ implies that
$\rct_{S_1}=\rct_{S_2}$, we also have
$\act(S_1)=\act(S_2)$. Intuitively, a state $S_1$ simulates a state
$S_2$ iff $(i)$ $S_1$ and $S_2$ coincide on all the active tasks
(i.e., the tasks $\tau_i$ s.t.\ $\rct_{S_1}(\tau_i)>0$), and $(ii)$ the
$\nat$ of each idle task is not larger in $S_1$ than in $S_2$. Let us
show that this preorder is indeed a simulation relation when we
consider a \emph{memoryless} scheduler (which is often the case in
practice):

\begin{theorem}
\label{theorem:idlesim}
Let $\tau$ be a set of sporadic tasks and let $\run$ be a
\emph{memoryless} (deterministic) scheduler for $\tau$ on $m$
processors. Then, $\idlesim$ is a simulation relation for
$A(\tau,\run)$.
\end{theorem}

\begin{proof}
  Let $S_1$, $S_1'$ and $S_2$ be three states in $\states{\tau}$
  s.t.\ $(S_1, S_1')\in E$ and $S_2\idlesim S_1$, and let us show that
  there exists $S_2'\in \states{\tau}$ with $(S_2,S_2')\in E$ and
  $S_2'\idlesim S_1'$.

  Since $(S_1, S_1')\in E$, there exists $\Sbar_1$ and $\tau'\subseteq
  \tau$ s.t.\ $S_1\xrightarrow{\tau'}\Sbar_1\xrightarrow{\run}S_1'$, by
  Definition~\ref{def:automaton}. Let $\Sbar_2$ be the (unique) state
  s.t.\ $S_2\xrightarrow{\tau'}\Sbar_2$, and let us show that
  $\Sbar_2\idlesim \Sbar_1$:
  \begin{enumerate}
  \item for all $\tau_i\in \tau'$:
    $\rct_{\Sbar_1}(\tau_i)=C_i=\rct_{\Sbar_2}(\tau_i)$. For all
    $\tau_i\not\in \tau'$:
    $\rct_{\Sbar_1}(\tau_i)=\rct_{S_1}(\tau_i)$,
    $\rct_{\Sbar_2}(\tau_i)=\rct_{S_2}(\tau_i)$, and, since
    $S_2\idlesim S_1$: $\rct_{S_1}(\tau_i)=\rct_{S_2}(\tau_i)$. Thus we
    conclude that $\rct_{\Sbar_1}=\rct_{\Sbar_2}$.
  \item Let $\tau_i$ be s.t.\ $\rct_{\Sbar_1}(\tau_i)=0$. Then, we must have
    $\tau_i\not\in \tau'$. In this case,
    $\nat_{\Sbar_1}(\tau_i)=\nat_{S_1}(\tau_i)$,
    $\nat_{\Sbar_2}(\tau_i)=\nat_{S_2}(\tau_i)$, and, since $S_2\idlesim
    S_1$, $\nat_{S_2}(\tau_i)\leq\nat_{S_1}(\tau_i)$. Hence,
    $\nat_{\Sbar_2}(\tau_i)\leq\nat_{\Sbar_1}(\tau_i)$. We conclude
    that for every $\tau_i$ s.t.\ $\rct_{\Sbar_1}(\tau_i)=0$:
    $\nat_{\Sbar_2}(\tau_i)\leq\nat_{\Sbar_1}(\tau_i)$
  \item By similar reasoning, we conclude that, for all $\tau_i$
    s.t.\ $\rct_{\Sbar_1}(\tau_i)>0$:
    $\nat_{\Sbar_1}(\tau_i)=\nat_{\Sbar_2}(\tau_i)$
  \end{enumerate}

  Then observe that, by Definition~\ref{def:simu}, $\Sbar_2\idlesim
  \Sbar_1$ implies that $\act(\Sbar_1)=\act(\Sbar_2)$. Let $\tau_i$ be
  a task in $\act(\Sbar_1)$, hence $\rct_{\Sbar_1}(\tau_i)>0$. In this
  case, and since $\Sbar_2\idlesim \Sbar_1$, we conclude that
  $\rct_{\Sbar_1}(\tau_i)=\rct_{\Sbar_2}(\tau_i)$ and
  $\nat_{\Sbar_1}(\tau_i)=\nat_{\Sbar_2}(\tau_i)$. Thus, since $\run$
  is memoryless by hypothesis, $\run(\Sbar_1)=\run(\Sbar_2)$, by
  Definition~\ref{def:sched}. Let $S_2'$ be the unique state
  s.t.\ $\Sbar_2\xrightarrow{\run}S_2'$, and let us show that
  $S_2'\idlesim S_1'$:
  \begin{enumerate}
  \item Since $\Sbar_2\idlesim\Sbar_1$, we know that
    $\rct_{\Sbar_1}=\rct_{\Sbar_2}$. Let $\tau_i$ be a task in
    $\run(\Sbar_1)=\run(\Sbar_2)$. By Definition~\ref{def:clocktick}:
    $\rct_{S_1'}(\tau_i)=\rct_{\Sbar_1}(\tau_i)-1$ and
    $\rct_{S_2'}(\tau_i)=\rct_{\Sbar_2}(\tau_i)-1$.  Hence,
    $\rct_{S_1'}(\tau_i)=\rct_{S_2'}(\tau_i)$. For a task
    $\tau_i\not\in\run(\Sbar_1)=\run(\Sbar_2)$, we have
    $\rct_{S_1'}(\tau_i)=\rct_{\Sbar_1}(\tau_i)$ and
    $\rct_{S_2'}(\tau_i)=\rct_{\Sbar_2}(\tau_i)$, again by
    Definition~\ref{def:clocktick}. Hence,
    $\rct_{S_1'}(\tau_i)=\rct_{S_2'}(\tau_i)$. We conclude that
    $\rct_{S_1'}=\rct_{S_2'}$.
  \item Let $\tau_i$ be a task s.t.\ $\rct_{S_1'}(\tau_i)=0$. By
    Definition~\ref{def:clocktick}:
    $\nat_{S_1'}(\tau_i)=\max\{0,\nat_{\Sbar_1}(\tau_i)-1\}$ and
    $\nat_{S_2'}(\tau_i)=\max\{0,\nat_{\Sbar_2}(\tau_i)-1\}$. However,
    since $\Sbar_2\idlesim \Sbar_1$, we know that
    $\nat_{\Sbar_1}(\tau_i)\leq\nat_{\Sbar_2}(\tau_i)$. We conclude
    that $\nat_{S_1'}(\tau_i)\leq \nat_{S_2'}(\tau_i)$.
  \item Let $\tau_i$ be a task s.t.\ $\rct_{S_1'}(\tau_i)>0$. By
    Definition~\ref{def:clocktick}:
    $\nat_{S_1'}(\tau_i)=\max\{0,\nat_{\Sbar_1}(\tau_i)-1\}$ and
    $\nat_{S_2'}(\tau_i)=\max\{0,\nat_{\Sbar_2}(\tau_i)-1\}$. Since
    $\rct_{S_1'}(\tau_i)>0$, we have $\rct_{\Sbar_1}(\tau_i)>0$ too,
    since $\rct$ can only decrease with time elapsing. Since
    $S_1\idlesim S_2$ we have also
    $\nat_{\Sbar_2}(\tau_i)=\nat_{\Sbar_1}(\tau_i)$. We conclude that
    $\nat_{S_1'}(\tau_i)=\nat_{S_2'}(\tau_i)$.
  \end{enumerate}

  To conclude the proof it remains to show that, if $S_2\idlesim S_1$
  and $S_1\in\fail_\tau$ then $S_2\in\fail_\tau$ too. Let $\tau_i$ be
  a task s.t
  $\lax_{S_1}(\tau_i)=\nat_{S_1}(\tau_i)-(T_i-D_i)-\rct_{S_1}(\tau_i)<0$. Since
  $S_2\idlesim S_1$: $\rct_{S_2}(\tau_i)=\rct_{S_1}(\tau_i)$, and
  $\nat_{S_2}(\tau_i)\leq\nat_{S_1}(\tau_i)$. Hence,
  $\lax_{S_2}(\tau_i)=\nat_{S_2}(\tau_i)-(T_i-D_i)-\rct_{S_2}(\tau_i)\leq
  \lax_{S_1}(\tau_i)<0$, and thus, $S_2\in\fail_\tau$.

\end{proof}

Note that Theorem \ref{theorem:idlesim} does {\em not} require the
scheduler to be work-conserving. Theorem \ref{theorem:idlesim} tells
us that any state where tasks have to wait until their next job
release can be simulated by a corresponding state where they can
release their job earlier, regardless of the specifics of the
scheduling policy as long as it is deterministic, predictable and
memoryless, which is what many popular schedulers are in practice,
such as preemptive DM or EDF.

\begin{figure}[ht]
\begin{center}
\includegraphics[scale=0.57]{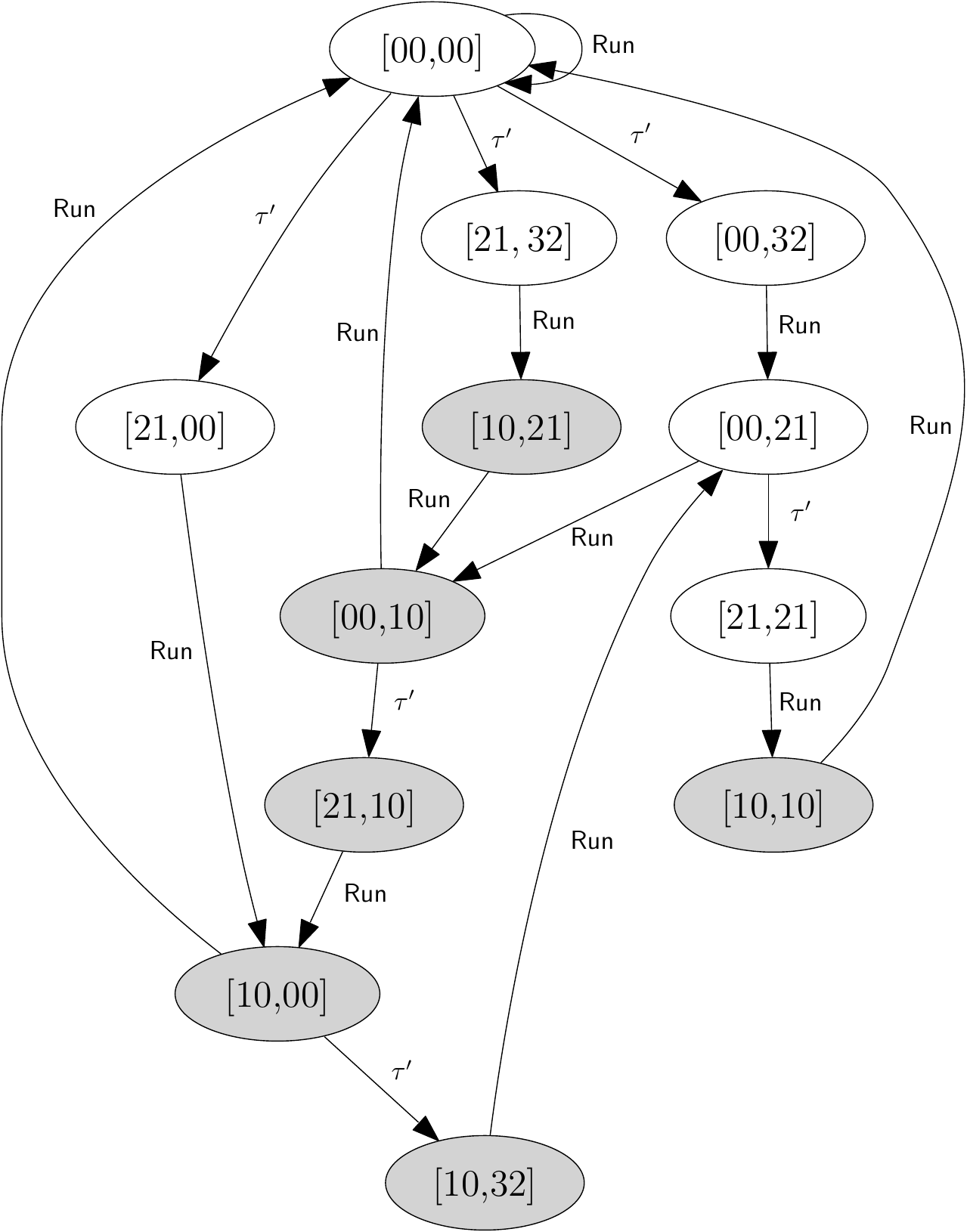}
\end{center}
\caption{Algorithm~\ref{algo:bfanti} exploits simulation relations to avoid exploring states needlessly. With $\idlesim$ on this small example, all grey states can be avoided as they are simulated by another state (e.g.\ $[00,21] \idlesim [10,21]$ and $[00,00] \idlesim [10,10]$).}
\label{fig:simulation}
\end{figure}

Figure~\ref{fig:simulation}, previously presented in Section~\ref{section:problem}, illustrates the effect of using $\idlesim$ with Algorithm~\ref{algo:bfanti}. If a state $S_1$ has been encountered previously and we find another state $S_2$ such that $S_1 \idlesim S_2$, then we can avoid exploring $S_2$ and its successors altogether. However, note that this does not mean we will never encounter a successor of $S_2$ as they may be encountered through other paths (or indeed, may have been encountered already).

\Section{Experimental results\label{section:experimental}}

We implemented both Algorithm~\ref{algo:bf} (denoted \emph{BF}) and Algorithm~\ref{algo:bfanti} (denoted \emph{ACBF} for ``antichain breadth-first'') in C++ using the STL and Boost libraries 1.40.0. We ran head-to-head tests on a system equipped with a quad-core 3.2~GHz Intel Core i7 processor and 12~GB of RAM running under Ubuntu Linux 8.10 for AMD64. Our programs were compiled with Ubuntu's distribution of GNU g++ 4.4.5 with flags for maximal optimization.

We based our experimental protocol on that used in~\cite{Baker2007}. We generated random task sets where task minimum interarrival times $T_i$ were uniformly distributed in $\{1,2,\ldots,T_{\max}\}$, task WCETs $C_i$ followed an exponential distribution of mean $0.35\,T_i$ and relative deadlines were uniformly distributed in $\{C_i,\ldots, T_i\}$. Task sets where $n \leq m$ were dropped as well as sets where $\sum_i C_i/T_i > m$. Duplicate task sets were discarded as were sets which could be scaled down by an integer factor. We used EDF as scheduler and simulated $m=2$ for all experiments. Execution times (specifically, used CPU time) were measured using the C \texttt{clock()} primitive.

\begin{figure}
\begin{center}
\includegraphics[scale=0.45]{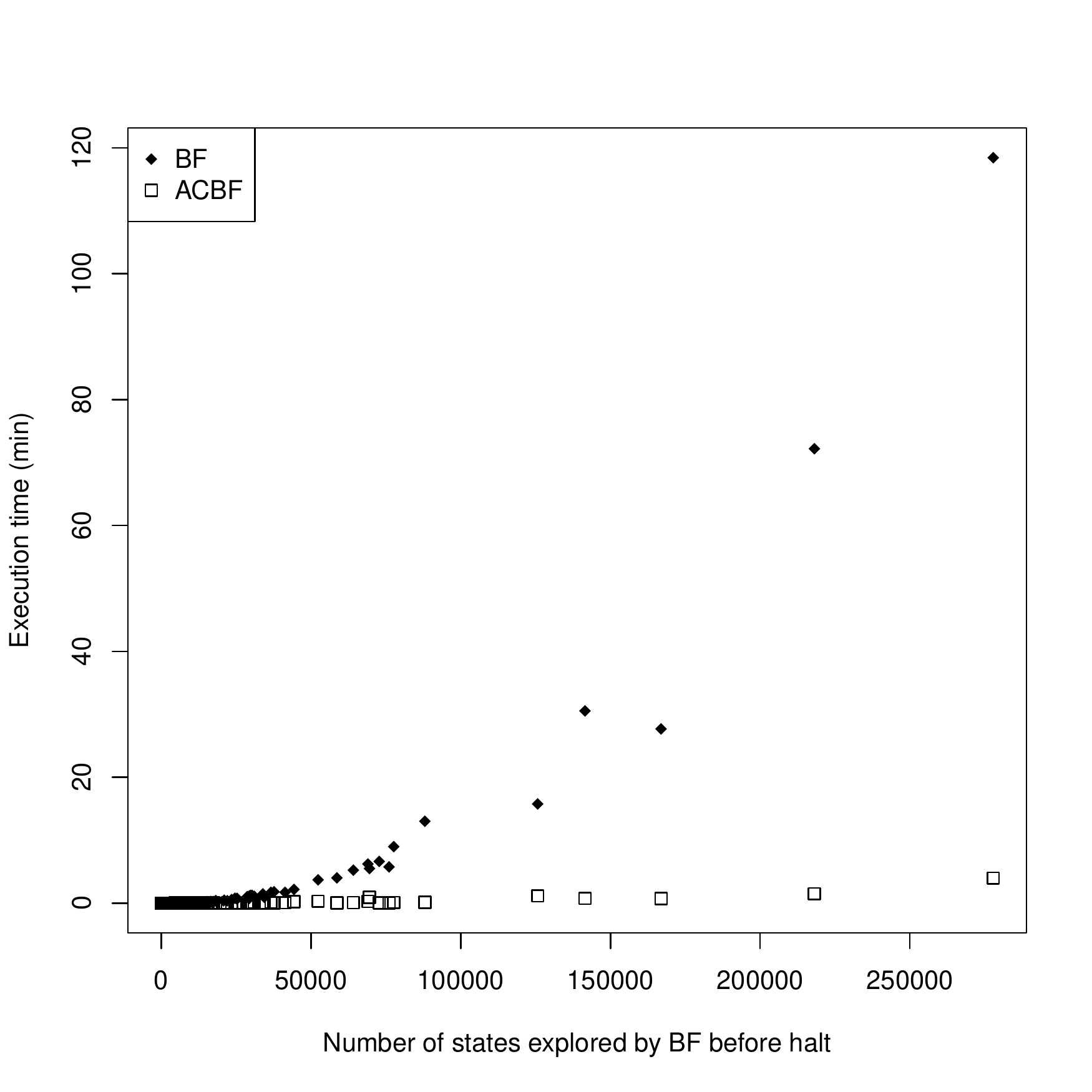}
\end{center}
\caption{States explored by BF before halt vs. execution time of BF and ACBF (5,000 task sets with $T_{\max} = 6$).}
\label{fig:states_vs_speed}
\end{figure}

Our first experiment used $T_{\max} = 6$ and we generated 5,000 task
sets following the previous rules (of which 3,240 were
EDF-schedulable). Figure~\ref{fig:states_vs_speed} showcases the
performance of both algorithms on these sets. The number of states
explored by BF before halting gives a notion of how big the automaton
was (if no failure state is reachable, the number is exactly the
number of states in the automaton that are reachable from the initial
state; if a failure state is reachable, BF halts before exploring the
whole system). It can be seen that while ACBF and BF show similar
performance for fairly small systems (roughly up to 25,000 states),
ACBF outperforms BF for larger systems, and we can thus conclude that
the antichains technique \emph{scales better}. The largest system
analyzed in this experiment was schedulable (and BF thus had to
explore it completely), contained 277,811 states and was handled in
slightly less than 2 hours with BF, whereas ACBF clocked in at 4
minutes.

\begin{figure}
\begin{center}
\includegraphics[scale=0.45]{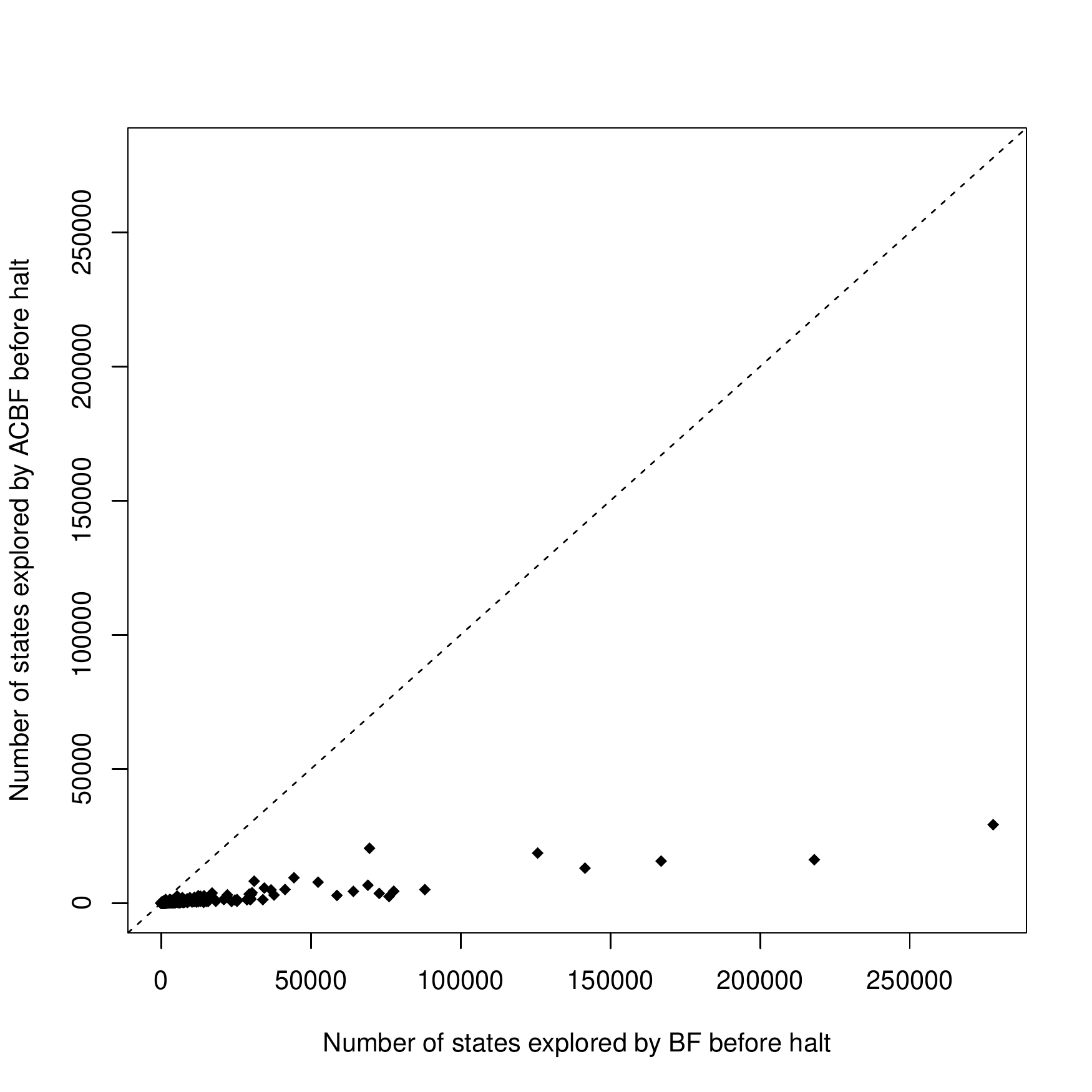}
\end{center}
\caption{States explored by BF before halt vs. states explored by ACBF before halt (5,000 task sets with $T_{\max} = 6$).}
\label{fig:bc_vs_ac_states}
\end{figure}

Figure~\ref{fig:bc_vs_ac_states} shows, for the same experiment, a
comparison between explored states by BF and ACBF. This comparison is
more objective than the previous one, as it does not account for the
actual efficiency of our crude implementations. As can be seen, the
simulation relation allows ACBF to drop a considerable amount of
states from its exploration as compared with BF: on average, 70.8\%
were avoided (64.0\% in the case of unschedulable systems which cause
an early halt, 74.5\% in the case of schedulable systems). This of
course largely explains the better performance of ACBF, but we must
also take into account the overhead due to the more complex
algorithm. In fact, we found that in some cases, ACBF would yield
worse performance than BF. However, to the best of our knowledge, this
only seems to occur in cases where BF took relatively little time to
execute (less than five seconds) and is thus of no concern in
practice.


\begin{figure}
\begin{center}
\includegraphics[scale=0.45]{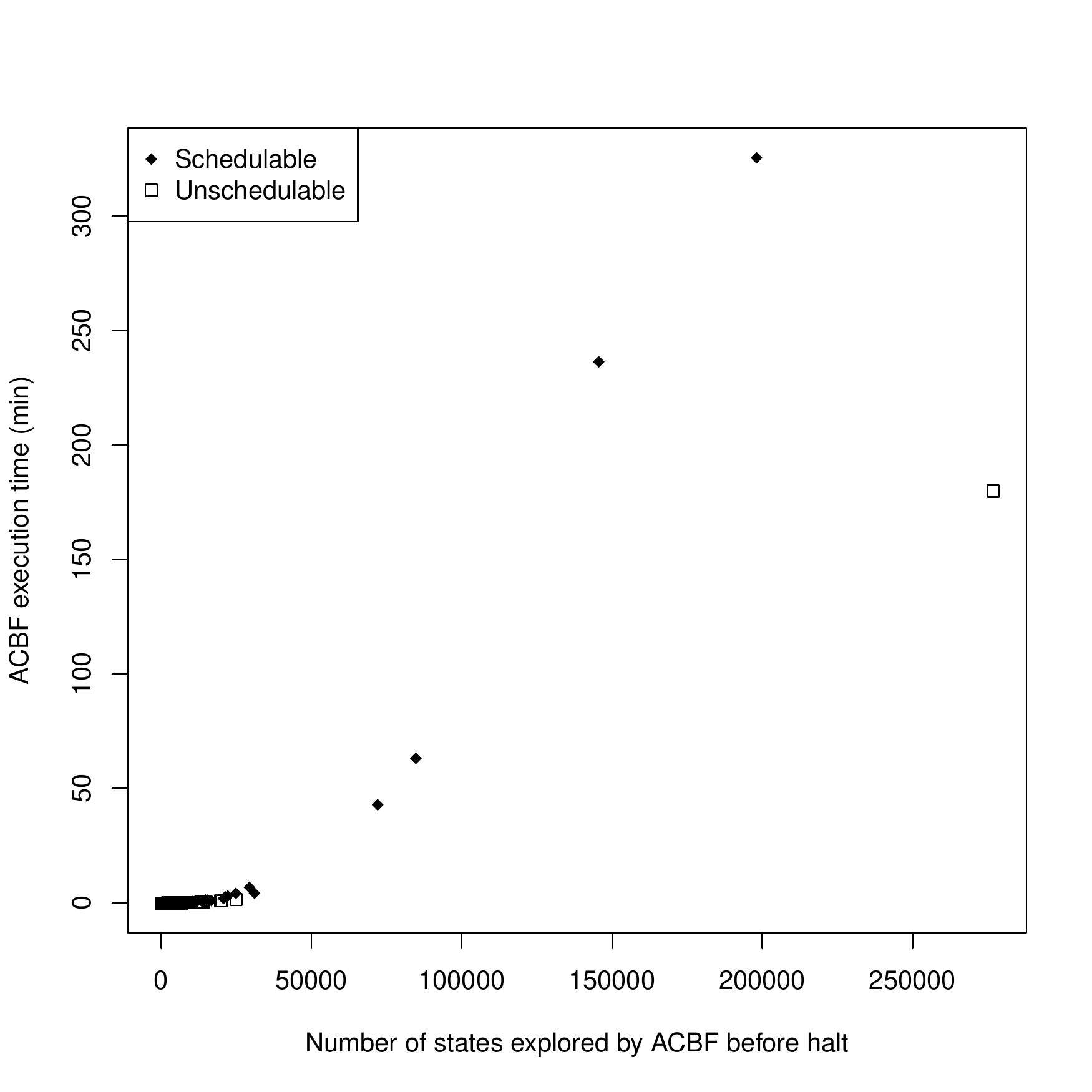}
\end{center}
\caption{States explored by ACBF before halt vs. ACBF execution time (5,000 task sets with $T_{\max} = 8$).}
\label{fig:ac_states_vs_time}
\end{figure}

Our second experiment used 5,000 randomly generated task sets using $T_{\max} = 8$ (of which 3,175 were schedulable) and was intended to give a rough idea of the limits of our current ACBF implementation. Figure~\ref{fig:ac_states_vs_time} plots the number of states explored by ACBF before halting versus its execution time. We can first notice the plot looks remarkably similar to BF in Figure~\ref{fig:states_vs_speed}, which seems to confirm the exponential complexity of ACBF which we predicted. The largest schedulable system considered necessitated exploring 198,072 states and required roughly 5.5 hours. As a spot-check, we ran BF on a schedulable system where ACBF halted after exploring 14,754 states in 78 seconds; BF converged after just over 6 hours, exploring 434,086 states.

Our experimental results thus yield several interesting observations. The number of states explored by ACBF using the idle tasks simulation relation is significantly less on average than BF. 
This gives an objective metric to quantify the computational performance gains made by ACBF wrt BF. In practice using our implementation, ACBF outperforms BF for any reasonably-sized automaton, but we have seen that while our current implementation of ACBF defeats BF, it gets slow itself for slightly more complicated task sets. However, we expect smarter implementations and more powerful simulation relations to push ACBF much further.

\Section{Conclusions and future work}
\label{section:conclusions}

We have successfully adapted a novel algorithmic technique developed by the formal verification community, known as antichain algorithms~\cite{DeWulf2006, Doyen2010}, to greatly improve the performance of an existing exact schedulability test for sporadic hard real-time tasks on identical multiprocessor platforms~\cite{Baker2007}. To achieve this, we developed and proved the correctness of a simulation relation on a formal model of the scheduling problem. While our algorithm has the same worst-case performance as a naive approach, we have shown experimentally that our preliminary implementation can still outperform the latter in practice.

The model introduced in Section~\ref{section:automaton} yields the added contribution of bringing a fully formalized description of the scheduling problem we considered. This allowed us to formally define various scheduling concepts such as memorylessness, work-conserving scheduling and various scheduling policies. These definitions are univocal and not open to interpretation, which we believe is an important consequence. We also clearly define what an execution of the system is, as any execution is a possibly infinite path in the automaton, and all possible executions are accounted for.

We expect to extend these results to the general Baker-Cirinei automaton which allows for arbitrary deadlines in due time. We chose to focus on constrained deadlines in this paper mainly because it simplified the automaton and made our proofs simpler, but we expect the extension to arbitrary deadlines to be fairly straightforward. We also only focused on developing \emph{forward simulations}, but there also exist antichain algorithms that use {\em backward simulations}~\cite{Doyen2010}. It would be interesting to research such relations and compare the efficiency of those algorithms with that presented in this paper.

The task model introduced in Section~\ref{section:problem} can be further extended to enable study of more complex problems, such as job-level parallelism and semi-partitioned scheduling. The model introduced in Section~\ref{section:automaton} can also be extended to support broader classes of schedulers. This was briefly touched on in~\cite{Baker2007}. For example, storing the previous scheduling choice in each state would allow modelling of non-preemptive schedulers.

It has not yet been attempted to properly optimize our antichain algorithm by harnessing adequate data structures; our objective in this work was primarily to get a preliminary ``proof-of-concept'' comparison of the performance of the naive and antichain algorithms. Adequate implementation of structures such as \emph{binary decision diagrams} \cite{Bryant1992} and \emph{covering sharing trees} \cite{Delzanno2004} should allow pushing the limits of the antichain algorithm's performance.

Antichain algorithms should terminate quicker by using coarser simulation preorders. Researching other simulation preorders on our model, particularily preorders that are a function of the chosen scheduling policy, is also key to improving performance. Determining the complexity class of sporadic task set feasability on identical multiprocessor platforms is also of interest, as it may tell us whether other approaches could be used to solve the problem.

\appendix

\Section{Proof of Lemma~\ref{lem:prop-bf-anti}}
\label{appendix:lemma}
In order to establish the lemma, we first show that, for any set $B$
of states, the following holds:
\begin{lemma}\label{lem:post}
  $\MaxElt{\fsim}{\Post{\MaxElt{\fsim}{B}}}=\MaxElt{\fsim}{\Post{B}}$.
\end{lemma}
\begin{proof}
  We first show that
  $\MaxElt{\fsim}{\Post{\MaxElt{\fsim}{B}}}\subseteq\MaxElt{\fsim}{\Post{B}}$. By
  def of $\MaxElt{\fsim}{B}$, we know that $\MaxElt{\fsim}{B}\subseteq
  A$. Moreover, $\mathsf{Succ}$ and $\MaxEltname{\fsim}$ are monotonic wrt
  set inclusion. Hence:
  $$
  \begin{array}{cl}
    &\MaxElt{\fsim}{B}\subseteq B\\
    \Rightarrow&\Post{\MaxElt{\fsim}{B}}\subseteq \Post{B}\\
    \Rightarrow&\MaxElt{\fsim}{\Post{\MaxElt{\fsim}{B}}}\subseteq \MaxElt{\fsim}{\Post{B}}
  \end{array}
  $$

  Then, we show that
  $\MaxElt{\fsim}{\Post{\MaxElt{\fsim}{B}}}\supseteq\MaxElt{\fsim}{\Post{B}}$. Let
  $S_2$ be a state in $\MaxElt{\fsim}{\Post{B}}$. Let $S_1\in B$ be a
  state s.t.\ $(S_1,S_2)\in E$. Since, $S_2\in\Post{B}$, $S_1$ always
  exists. Since $S_1\in B$, there exists $S_3\in\MaxElt{\fsim}{B}$
  s.t.\ $S_3\fsim S_1$. By Definition~\ref{def:simu}, there is
  $S_4\in\Post{\MaxElt{\fsim}{B}}$ s.t.\ $S_4\fsim S_2$. To conclude,
  let us show \textit{per absurdum} that $S_4$ is maximal in
  $\Post{\MaxElt{\fsim}{B}}$. Assume there exists
  $S_5\in\Post{\MaxElt{\fsim}{B}}$ s.t.\ $S_5\fsimstrict S_4$. Since
  $\MaxElt{\fsim}{B}\subseteq A$, $S_5$ is in $\Post{B}$
  too. Moreover, since $S_4\fsim S_2$ and $S_5\fsimstrict S_4$, we
  conclude that $S_5\fsimstrict S_2$. Thus, there is, in $\Post{B}$
  and element $S_5\fsimstrict S_2$. This contradict our hypothesis
  that $S_2\in\MaxElt{\fsim}{\Post{B}}$.
\end{proof}

    \begin{figure*}[t!h!]
    \begin{center}
    \emph{Induction hypotesis} we assume
    that $\tR_{k-1}=\MaxElt{\fsim}{R_k}$. Then:
    \[
    \begin{array}{cll}
      &\tR_k\\
      =&\MaxElt{\fsim}{\tR_{k-1}\cup\Post{\tR_{k-1}}}&\textrm{By def.}\\
      =&\MaxElt{\fsim}{\MaxElt{\fsim}{\tR_{k-1}}\cup\MaxElt{\fsim}{\Post{\tR_{k-1}}}}&\textrm{by~(\ref{eq:cup})}\\
      =&\MaxElt{\fsim}{\MaxElt{\fsim}{\MaxElt{\fsim}{R_{k-1}}}}\cup\MaxElt{\fsim}{\Post{\MaxElt{\fsim}{R_{k-1}}}}&\textrm{By I.H.}\\
      =&\MaxElt{\fsim}{\MaxElt{\fsim}{R_{k-1}}\cup\MaxElt{\fsim}{\Post{R_{k-1}}}}&\textrm{By Lemma~\ref{lem:post}}\\
      =&\MaxElt{\fsim}{R_{k-1}\cup\Post{R_{k-1}}}&\textrm{By~(\ref{eq:cup})}\\
      =&\MaxElt{\fsim}{R_k}&\textrm{By def.}
    \end{array}
    \]
    \caption{Inductive case for Lemma \ref{lemma:antichain}.}
    \label{fig:proof}
    \end{center}
    \end{figure*}

Then, we are ready to show that:
\begin{lemma}
\label{lemma:antichain}
  Let $A$ be an automaton and let $\fsim$ be a simulation relation for
  $A$. Let $R_0,R_1,\ldots$ and $\tR_0,\tR_1,\ldots$ denote
  respectively the sequence of sets computed by
  Algorithm~\ref{algo:bf} and Algorithm~\ref{algo:bfanti} on
  $A$. Then, for all $i\geq 0$: $\tR_i=\MaxElt{\fsim}{R_i}$.
\end{lemma}
\begin{proof}
  The proof is by induction on $i$. We first observe that for any pair
  of sets $B$ and $C$, the following holds: 
  \begin{eqnarray}
    &\begin{array}{cl}
      &\MaxElt{\fsim}{B\cup
        C}\\
      =&\MaxElt{\fsim}{\MaxElt{\fsim}{B}\cup \MaxElt{\fsim}{C}}
    \end{array}
    \label{eq:cup}
  \end{eqnarray}

  \begin{description}
  \item[Base case $i=0$] Clearly, $\MaxElt{\fsim}{R_0}=R_0$ since
    $R_0$ is a singleton. By definition $\tR_0=R_0$.
   \item[Inductive case $i=k$] See Figure~\ref{fig:proof}.
    
  \end{description}
\end{proof}

\paragraph{Acknowledgment.} We thank Phan Hiep Tuan for identifying mistakes in Definition~\ref{def:simu} and Theorem~\ref{theorem:idlesim}.

%
%
%
%
%
%
%
%
%
%
%
%

\bibliographystyle{latex8}
\bibliography{antichains}

\end{document}